\documentclass[aps,prx,twocolumn,superscriptaddress,noeprint,longbibliography, nofootinbib]{revtex4-1}

\usepackage{graphicx}% Include figure files
\usepackage{bm}% bold math
\usepackage{physics}
\usepackage{xcolor}
\usepackage{enumitem}
\usepackage{amsmath, amssymb}
\usepackage[normalem]{ulem}
\usepackage{natbib}
\usepackage{bbm}
\usepackage{comment}
\usepackage{mathrsfs}
\usepackage{amsmath}
\usepackage{amsthm}

\theoremstyle{definition} 
\newtheorem{theorem}{Proposition}

\usepackage[colorlinks=true]{hyperref}
\hypersetup{citecolor = blue}

\newcommand{\im}{\textrm{i}}

\newcommand{\vphi}{\varphi}

\newenvironment{centeq}
    {\begin{equation} \begin{aligned}}
    {\end{aligned} \end{equation}}

\begin{document}

\title{Anyon Crystals and Hall Crystals in a Periodic Potential}

\author{Sayak Bhattacharjee}
\email{sayakbhattacharjee@stanford.edu}
\affiliation{Leinweber Institute for Theoretical Physics, Stanford University, Stanford, CA 94305, USA}

\author{Julian May-Mann}
\affiliation{Leinweber Institute for Theoretical Physics, Stanford University, Stanford, CA 94305, USA}

\author{Srinivas Raghu}
\affiliation{Leinweber Institute for Theoretical Physics, Stanford University, Stanford, CA 94305, USA}

\begin{comment}

\end{comment}

\begin{abstract}

We obtain integer and fractional quantum Hall crystals as ground states of a two-dimensional electron system subject to a strong perpendicular magnetic field and a periodic potential. For certain fractional states, we show that the Hall crystal can constitute an \textit{anyon crystal}, with a periodic ordering of well-defined anyons. We find  that the latter states can be stabilized at odd denominator Landau level filling fractions when Landau level mixing is sufficiently weak, and near half-filling of the underlying lattice. These phases are obtained from a mean-field analysis of an effective lattice model of bosons attached to an odd number of flux quanta, which transmutes their statistics to that of electrons. In boson coordinates, the Hall crystal is a supersolid: a superfluid with charge order.  Under strong interactions, vortex-anti-vortex pairs spontaneously nucleate in the supersolid, realizing a crystalline state of anyons.

\end{abstract}

\maketitle

%\tableofcontents
%\subsection*{Introduction}

\textit{Introduction}\:---\:Quantum phases of matter in which  topological order coexists with broken symmetry have been of significant interest~\cite{sondhi1993skyrmions, tevsanovic1989hall}. 
One such phase, the Hall crystal, has attracted renewed interest, following experiments in two-dimensional materials~\cite{lu2024fractional, aronson2025displacement}. The Hall crystal (HC) is characterized by a quantized Hall conductivity~\cite{footnote_pin} and a periodic modulation of the density that breaks translational symmetry~\cite{kivelson1986cooperative, kivelson1987cooperative, tevsanovic1989hall, halperin1986compatibility, murthy2000hall, paul2025designing, balents1996spatially, narevich2001hamiltonian, may2026composite}. Both integer and fractional Hall crystals are viable possibilities, identified by their Hall conductivity.

We consider lattice analogs of Hall crystals where the discrete translation symmetry of the lattice is spontaneously broken. These lattice Hall crystals are interesting in their own right, since lattice quantum Hall states exhibit distinctive phenomena that are absent in continuum Landau levels (LL)~\cite{hofstadter1976energy}.

To study such states, we utilize composite boson theory~\cite{zhang1989effective, zhang1992chern}, wherein each boson is `attached' to an odd number of flux quanta, imparting fermionic statistics. Thus, electrons in the periodic potential are described by an extended Bose-Hubbard (BH) model~\cite{fisher1989boson} with additional non-local interactions mediated by gauge fluctuations. For fixed rational magnetic filling, we find that the boson model realizes superfluid, supersolid~\cite{leggett1970can, prokof2005supersolid, andreev1969quantum, matsuda1970off}, and bosonic Wigner-Mott (WM)~\cite{hubbard1978generalized} insulating phases, which correspond to quantum Hall state (HS), Hall crystal, and electronic WM insulating phases, respectively~\cite{balents1996spatially, zhang1992chern, may2026composite}.

\begin{figure}
    \centering    \includegraphics[width=\columnwidth]{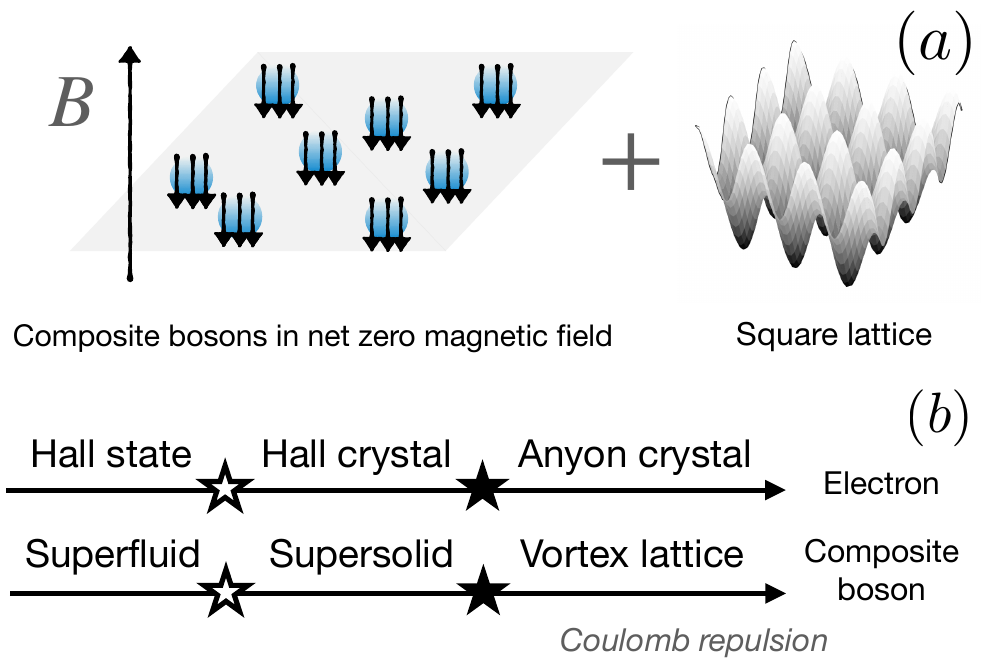}
    \caption{\textbf{(a)} Schematic of the two dimensional electron system in the presence of a square lattice in composite boson coordinates. We depict composite bosons (blue circles) `attached' to an odd number of flux quanta; say, 3 for $\nu=1/3$. \textbf{(b)} Schematic phase diagrams in electron and composite boson coordinates plotted on an axis denoting increasing ratio of Coulomb interactions to kinetic energy, towards the right (at fixed LL mixing). The anyon crystal at strong interactions is obtained at weak LL mixing and near half filling, say, a lattice density of $9/16$ and $\nu=1/3$. In this case, the Hall state and Hall crystal are separated by a continuous phase transition (white star) while the transition between the Hall crystal and the anyon crystal is first-order 
    (black star).}
    \label{fig:fig1}
\end{figure}

Additionally, in certain parameter regimes, we find that the bosonic model hosts a spontaneously formed lattice of vortex-anti-vortex pairs. In electronic coordinates, this is an \textit{anyon crystal} (AC)~\cite{lee1989anyon, may2026composite, pichler2026microscopic, kuhlenkamp2025robust}, an exotic type of fractional Hall crystal that hosts a lattice of well-defined anyons, in addition to ordinary density modulations. Importantly, these crystallized anyons are not produced by doping quasiparticles into an existing fractionalized state. Instead, strong interactions spontaneously nucleate oppositely charged anyon pairs.   

We study this problem as a function of the ratio of kinetic energy to electronic interactions and the ratio of interactions to the cyclotron frequency $\omega_c$ (the LL mixing parameter). Generically, for strong electronic interactions, a WM insulator is obtained, which upon increasing kinetic energy melts into the HS through an intermediate HC state. At lattice fillings $(\bar{\rho})$ near half filling and sufficiently weak LL mixing, the WM insulator is replaced by an AC. This corresponds to a region of parameter space where the ground state is the result of dominant statistical interactions over the kinetic energy or electronic interactions.

\textit{Lattice Hamiltonian}\:---\:In the continuum, the first-quantized Hamiltonian $H$ of the two-dimensional electronic system subjected to a strong perpendicular magnetic field $(B)$ and a periodic potential is given by,
\begin{equation}\label{electronic_Hamiltonian}
H=\frac{1}{2m}\sum_{i}\left(\bm{p}_i-\bm{A}\right)^2+\sum_{i<j}\mathcal{V}(\bm{r}_{ij})+\sum_i V_{\textrm{lat}}(\bm{r}_i),
\end{equation}
where $\mathcal{V}$ denotes the repulsive interaction and the third term denotes the lattice potential. ($\bm{r}_{ij}=\bm{r}_i-\bm{r}_j$, $\hbar=c=e=1$.) Specifically, we consider a square lattice (with unit lattice constant) and screened Coulomb repulsion (limited to nearest neighbors). The background gauge field is $\bm{A}$, so that $B=\bm{\nabla}\times \bm{A}$. We denote the first and second term in Eq.~\ref{electronic_Hamiltonian} by $H_{\textrm{kin}}$ and $H_{\textrm{int}}$ respectively for future convenience. We shall also limit to magnetic fillings $\nu\: (=\bar{\rho}\phi_0/B)$, where $\nu^{-1}$ is an odd integer, and $\phi_0=2\pi$ is the quantum of flux.  

To switch from electronic to composite boson coordinates in Eq.~\ref{electronic_Hamiltonian}, we introduce a coupling to a statistical gauge field $\bm{a}(\bm{r})$ in the Hamiltonian. It is convenient to consider a shifted gauge field, $\bm{a}'=\bm{A}-\bm{a}$, so that the kinetic energy of the bosons is $H_{\textrm{kin}}=(1/2m)\sum_i\left(\bm{p}_i-\bm{a}'(\bm{r}_i)\right)^2$. The interaction and lattice potential terms of the Hamiltonian preserve their form. The bosons have fermionic statistics, which is implemented by a Gauss' law enforcing flux attachment, 
\begin{equation}\label{eq:flux_attachment}
    b'(\bm{r}) =\bm{\nabla}\times \bm{a}'(\bm{r})=\phi_0\nu^{-1}\left[\bar{\rho}-\rho(\bm{r})\right]=\phi_0\nu^{-1}\delta\rho(\bm{r}) 
\end{equation}
where $\rho(\bm{r})=\sum_j\delta^{2}(\bm{r}-\bm{r}_j)$. One smears the statistical flux $b'$ in space, allowing for a mean-field analysis~\cite{zhang1989effective}. In the superfluid, when the modulating density $\langle \delta\rho\rangle =0$, the bosons feel no magnetic fluxes ($\langle b'\rangle =B-\langle b\rangle =0$), where $\langle \cdot \rangle$ denotes the ground-state expectation value of the operator. In fact, the bosons feel no net flux in all of the phases obtained at fixed odd denominator $\nu$.

We integrate out the statistical gauge field using the flux attachment constraint [\textcolor{blue}{SM}]. We do this in Coulomb gauge ($\bm{\nabla}\cdot \bm{a}'=0$). Discretizing the continuum Hamiltonian, we obtain a boson lattice Hamiltonian $H_{\textrm{eff}}=H_{\textrm{BH}}+H_{\textrm{int}}+H_{\textrm{st-int}}$. $H_{\textrm{BH}}$ is the Bose-Hubbard Hamiltonian, $H_{\textrm{int}}$ represents finite-range electronic interactions and $H_{\textrm{st-int}}$ is the statistical interaction:
\begin{align}
\nonumber  H_{\textrm{BH}}=&-t\sum_{\langle ij\rangle }[\varphi^\dagger_i\varphi_j + \textrm{H.c.}]+\frac{U}{2}\sum_i \rho_i(\rho_i-1),\\
\nonumber H_{\textrm{int}}=&\: V\sum_{\langle ij\rangle}\rho_i\rho_j,\\
\nonumber  H_{\textrm{st-int}}=&-\frac{\omega_c}{\nu}\bigg[\sum_{i<j}\ln(r_{ij})\delta \rho_i\delta \rho_j+\frac{1}{2\pi\bar{\rho}}\sum_{i}\delta \rho_i (\bm{\mathcal{A}}_i)^2+ \\
 &\frac{\nu}{2\pi\bar{\rho}}\sum_{i}\bm{J}_i \cdot \bm{\mathcal{A}}_i\bigg],\label{eq:ham_CS} 
\end{align}
where $\varphi$ $(\varphi^\dagger)$ is the boson annihilation (creation) operator, $\rho=\varphi^\dagger \varphi$ the boson density, and $t$  the effective nearest-neighbor boson hopping amplitude. The density-dependent potential is denoted by $\bm{\mathcal{A}}_i = \sum_{j}\bm{g}(\bm{r}_{ij})\delta\rho_j$, where $\bm{g}(\bm{r})=\hat{z}\times (\bm{r}/r^2)$ is the Biot-Savart kernel. The current operator $J_i^\alpha=(j_i^\alpha)_p+ (j_{i-e_\alpha}^\alpha)_p$, where $(j_i^\alpha)_p = -\textrm{i}\varphi^\dagger_i \varphi_{i+e_\alpha}+\textrm{h.c.}$ is the paramagnetic current, $\alpha=x,y$ and $e_\alpha$ is the unit vector. This definition of the current operator ensures that the current-density interaction (third term in $H_{\textrm{st-int}}$) preserves inversion symmetry on the square lattice.

The logarithmic repulsion---the two-dimensional Coulomb potential related to the `statistical' photon---arises from the diamagnetic term $\propto \rho a'^2$, as does the three-body term (second term in $H_{\textrm{st-int}}$). It lifts the superfluid phonon to the cyclotron frequency at zero momentum, consistent with Kohn's theorem~\cite{zhang1992chern}. The current-density interaction induces a circulating current around a localized charge density with a strength inversely related to the distance from the charge, as expected from attaching magnetic fluxes to the bosons. Since it breaks time-reversal symmetry, it can induce vortex lattice ground-states, which correspond to ACs.

\textit{Mean-field calculation}\:---\:For the mean-field calculation, we work in the canonical ensemble. We employ the Gutzwiller variational wavefunction~\cite{krauth1992gutzwiller}. The wavefunction is a product state given by, 
\begin{equation}\label{gutzwiller}
    \ket{\psi\{c_{im}\}}=\bigotimes_i \sum_{m=0}^Mc_{im}\ket{i,m},
\end{equation}
where $\ket{i,m}$ is the Fock state at site $i$ with $m \: (\leq M)$ bosons. We minimize the energy $\mel{\psi}{H_{\textrm{eff}}}{\psi}$ over a set of complex variational parameters $\{c_{im}\}$ subject to the normalization constraints $\sum_m |c_{im}|^2=1$ for all $i$. For simplicity, in this work, we shall primarily focus on the hard-core limit of the bosons, which corresponds to $U\rightarrow \infty$ and $M=1$. Broadly, our results do not change when $U$ is finite.

Before analyzing the various Hall states, it is useful to consider the $\omega_c\rightarrow 0$ limit where the magnetic field is absent. In this limit, the model (Eq.~\ref{eq:ham_CS}) is the extended BH model with nearest-neighbor interactions. For hard-core bosons and at half-filling $(\bar{\rho}=1/2)$, a superfluid transitions to a checkerboard WM insulator without an intermediate supersolid phase via a first-order transition~\cite{kimura2011gutzwiller}. At small but finite $\omega_c$, we find that this phase structure remains unaltered [\textcolor{blue}{SM}]. Upon relaxing the hard-core constraint, an intermediate supersolid is obtained at half-filling and $\omega_c=0$~\cite{iskin2011route, ohgoe2012ground}. This evolves into a checkerboard HC at finite $\omega_c$. 

In addition to the phases discussed above that extend from known phases of the extended BH model, there is also potential for new types of crystals at finite $\omega_c$ that do not have direct analogs in time-reversal symmetry preserving BH models. Specifically, due to the current-density interaction, it is possible to have crystals with spontaneously formed vortices and anti-vortices. In the original electronic problem, such a crystal is a crystal of anyons.

While ACs are a generic possibility in strongly interacting quantum Hall systems, ACs can be disallowed in the model we are studying due to geometric constraints arising from the Biot-Savart kernel $\bm{g}(\bm{r})$. Namely, we find that within the Gutzwiller mean-field ansatz, the three-body interaction and current-density interaction have zero expectation value for any charge density that is inversion symmetric about every point on the lattice. This is easily observed by inverting the potential $\bm{\mathcal{A}}_i=\sum_j \bm{g}(\bm{r}_{ij})\delta\rho_j$ about a site $i$. For crystals with such a charge density, the vortex lattice cannot arise as a ground-state at rational $\nu$, because the coefficients of the variational wavefunction can be chosen to be real with impunity. 

Thus, to obtain an AC, one requires a crystal that does not have inversion symmetry about every point in the unit cell. This can be stabilized by the three-body interaction, which favors anisotropic charge densities. To account for this, we consider a somewhat large $4\times 4$ super-unit cell on the square lattice, which is sufficient to capture the relevant physics. We will further restrict our attention to filling $\bar{\rho}=9/16$ and fractional LL filling $\nu=1/3$, where all of the phases of interest are present. Results for other lattice and magnetic fillings are provided in [\textcolor{blue}{SM}]. 

The variational energies are evaluated by resolving the long-ranged interactions onto the unit cell using the Ewald summation method~\cite{bonsall1977some} [\textcolor{blue}{SM}]. To determine the ground states, we seed with $\sim20$ initial states, which involve both random and symmetric ansatze, as well as vortex lattices.

\begin{figure}
    \centering
    \includegraphics[width=1\linewidth]{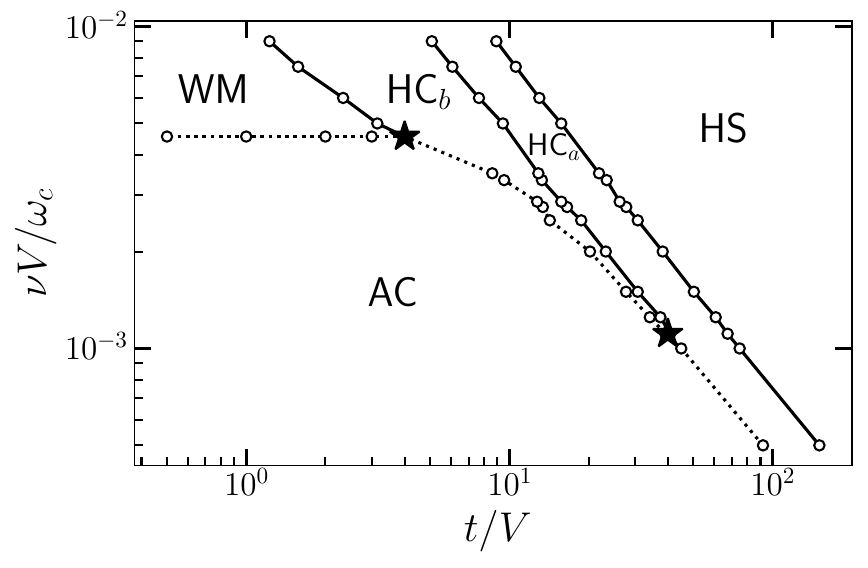}
    \caption{Phase diagram of the Hamiltonian in Eq.~\ref{eq:ham_CS} in the $t/V-\nu V/\omega_c$ plane on a log-log plot for $\bar{\rho}=9/16$, $\nu=1/3$ and $U\rightarrow \infty$.  HC$_a$ and HC$_b$ denote Hall crystals---HC$_a$ is a checkerboard Hall crystal, while $\textrm{HC}_b$ has a super-unit cell of $2\times 2$, just as the Wigner-Mott (WM) insulator. HS denotes the fractional Hall state. The crystals are sketched in Fig.~\ref{fig:crystal_orders}(a)-(d). The stars denote critical points. The dotted (solid) lines are first-order (continuous) transitions.}
    \label{fig:phase_diagram}
\end{figure}

\begin{figure*}
    \centering
    \includegraphics[width=1\linewidth]{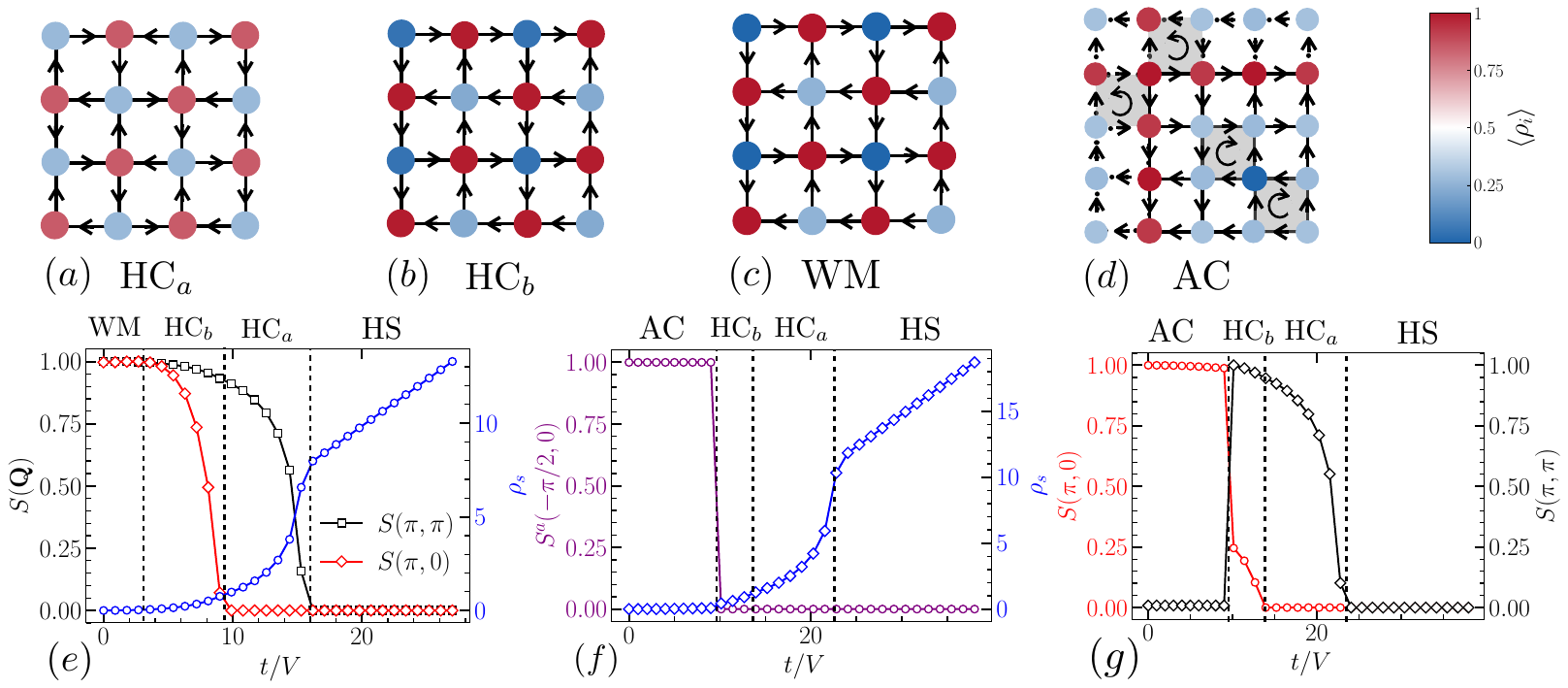}
    \caption{\textbf{(a)}, \textbf{(b)} and \textbf{(c)} depict schematic super-unit cells of the Hall crystals ($\textrm{HC}_a$, $\textrm{HC}_b$) and the Wigner-Mott insulator at $\bar{\rho}=9/16$ and $\nu=1/3$ respectively. The site densities are colored as per the color bar. In the WM insulator, the densities are $1, 1/4$ and $0$. The arrows denote the direction of the orbital currents at the center of the bonds. \textbf{(d)} depicts a schematic $5\times 5$ cell of the AC at $\bar{\rho}=9/16$ and $\nu=1/3$. The grey plaquettes indicate one of two candidate locations of the quasiholes and quasiparticles respectively, with the circular arrow denoting the vorticity. An energetically degenerate location for each pair is $\pi/2$ rotated with respect to the inversion center at the density extrema. \textbf{(e)} Order parameters at $\bar{\rho}=9/16$, $\nu V/\omega_c=0.005$ and $\nu=1/3$ [$S(\pi,0)$ (red, diamond), $S(\pi,\pi)$ (black, square) and $\rho_s$ (blue, circle)] \textbf{(f)} and \textbf{(g)}: Order parameters at $\nu V/\omega_c=0.0033$ and $\nu=1/3$ [$S^a(-\pi/2,0)$ (purple, circle), $\rho_s$ (blue, diamond), $S(\pi,0)$ (red, circle) and $S(\pi,\pi)$ (black, diamond)]. The structure factors are normalized with respect to their maximum value.} 
    \label{fig:crystal_orders}
\end{figure*}

\textit{Results}\:---\:In Fig.~\ref{fig:phase_diagram}, we plot the overall phase diagram for $\bar{\rho}=9/16$ in the $t/V-\nu V/\omega_c$ plane. The $(t/V)-$axis represents the ratio of the kinetic energy to the strength of screened Coulomb repulsion. At fixed $\nu$, the ($\nu V/\omega_c)-$axis tunes LL mixing---the relative strength of Coulomb repulsion to statistical interactions. 

First consider a line cut at fixed Landau level mixing corresponding to $\nu V/\omega_c=0.005$. We obtain four phases as a function of $t/V$; these are the fractional quantum Hall state, two Hall crystals, denoted by HC$_a$ and HC$_b$, and a Wigner-Mott insulator. The WM insulator and HC$_b$ have super-unit cells of size $2\times 2$, while HC$_a$ is of the checkerboard type (with a super-unit cell of size $2\times 1$). These crystals are depicted in Fig.~\ref{fig:crystal_orders}(a)-(c). To characterize these phases, we compute the bare superfluid stiffness $\rho_s$ (defined explicitly in [\textcolor{blue}{SM}]) and structure factors $S(\bm{Q})=(1/4^2)\sum_{i,j}e^{\textrm{i}\mathbf{Q}\cdot (\bm{r}_i-\bm{r}_j)}\langle \rho_i\rangle\langle \rho_j\rangle$ at wave-vector $\mathbf{Q}$. The WM insulator and HC$_b$ exhibit a non-zero $S(\pi,0)$ while $S(\pi,\pi)$ is non-zero also in the checkerboard HC$_a$. In Fig.~\ref{fig:crystal_orders}(f), we plot these order parameters as a function of $t/V$. We find that the phases are separated by continuous phase transitions.

Next consider a line cut at \textit{weaker} Landau level mixing, say $\nu V/\omega_c=0.0033$. Here, we again obtain a similar phase diagram; however, the WM insulator at small kinetic energy is replaced by an anyon crystal. In this case, we do not obtain a WM insulator even at small $t/V$ (at least upto $t/V\gtrsim 10^{-6})$. This may be understood as due to the stronger statistical interactions as compared to the screened Coulomb repulsion.  Increasing $t/V$ melts the AC first to HC$_b$, which transitions to HC$_a$ and then melts into the fractional HS. 

The AC obtained has an ornate density and phase pattern (see Fig.~\ref{fig:crystal_orders}(d)). Each super-unit cell of the AC has two vortices (quasiholes) and two anti-vortices (quasiparticles), with vorticity $\pm 2\pi$ (charge $\mp 1/3$), so that the total vorticity (charge) is zero. The overall charge neutrality of the anyons is guaranteed because this AC 
is realized at an odd denominator LL filling fraction. The AC density has two inversion centers at the density extrema, and the anyons sit in plaquettes pinned to these sites, as expected from  flux attachment. From the saddle-point of the statistical interaction energy, one indeed finds that the local vorticity 
is $\propto \delta\rho(\bm{r})$.

The transition between the WM insulator and the AC is first order. Approaching from the AC, this transition may be understood as through spontaneous annihilation of vortex-anti-vortex pairs. We plot the order parameters as a function of $t/V$ in Fig.~\ref{fig:crystal_orders}(g) and (h). To characterize the AC, we compute an anyon structure factor---$S^{a}(-\pi/2,0)=(1/4^2)\sum_{\alpha,\beta}e^{\textrm{i}(-\pi/2,0)\cdot (\bm{r}_\alpha-\bm{r}_\beta)}\omega_\alpha\omega_\beta$, where $\alpha,\beta$ label the plaquettes. The loss in chemical pressure corresponding to the anyon density  realizes a jump in the anyon structure factor at the transition, as shown in Fig.~\ref{fig:crystal_orders}(g).

As shown in Fig.~\ref{fig:phase_diagram}, for sufficiently small $\nu V/\omega_c$, the $\textrm{HC}_b$ phase disappears (indicated by the critical point at $t/V\approx 42$, $\nu V/\omega_c\approx 0.001$). The approximately linear phase boundary in the log-log plot between HL and $\textrm{HC}_a$ in Fig.~\ref{fig:phase_diagram} suggests that the location of the transition is primarily controlled by the competition between $t$ and $\omega_c/\nu$, even though the checkerboard order of HC$_a$ is due to the nearest-neighbor interaction. Assuming that this transition survives quantum fluctuations, this transition is expected to lie in the Ising universality class. A similar argument can be made for the transition between $\textrm{HC}_a$ and $\textrm{HC}_b$.

\textit{Orbital currents}\:---\:Even if the mean-field crystalline ground-states are free of vortices, they break time-reversal symmetry. This can be directly observed from the presence of orbital currents $(\bm{j}_i)_d$ in the crystalline ground states. The gauge-invariant current on the lattice at the center of a bond assigned to site $i$ is $\bm{j}_i=(\bm{j}^p)_i+(\bm{j}^d)_i$, where the diamagnetic current may be written as
\begin{equation}
    (\bm{j}^d)_i=\frac{\phi_0}{2\pi\nu}\frac{1}{2}(\rho_i+\rho_{i+e_\alpha})\sum_\alpha \frac{1}{Q_\alpha^2 }\hat{z}\times \nabla(\rho_\alpha)_i,
\end{equation}
and $\nabla$ denotes the forward difference operator on the lattice [\textcolor{blue}{SM}]. This formula applies to crystals with Fourier weight over a finite set of wavevectors $\bm{Q}_\alpha$, such that their density can be written as $\rho_i = \bar{\rho} +\sum_\alpha(\rho_\alpha)_i$. To obtain this expression, we again used the flux attachment constraint (Eq.~\ref{eq:flux_attachment}). We note that the paramagnetic currents are only non-zero in the AC, and so the total current is equal to the diamagnetic contribution elsewhere. We plot the currents in the crystals in Fig.~\ref{fig:crystal_orders}(a)-(d). For $\textrm{HC}_a$ and $\textrm{HC}_b$, we find that the currents realize $d-$density wave and modulated stripe current order respectively~\cite{chakravarty2001hidden, nayak2000density} (see Fig.~\ref{fig:crystal_orders}(a) and (b)).

\textit{Discussion}\:---\:The quantum Hall regime admits a variety of effective descriptions for understanding its emergent behavior~\cite{laughlin1983anomalous, zhang1989effective, jain1989composite, halperin1984statistics}. In this work, we have adopted the composite boson framework because it provides a unified perspective on both quantum Hall states and Hall crystal phases, interpreting them as manifestations of superfluidity and supersolidity, respectively. Additionally, within the lattice formulation developed here, the composite boson approach highlights the role of geometric frustration, demonstrating how Hall crystal and anyon crystals emerge from the interplay between frustration arising from short-range repulsive and long-range statistical interactions.

The anyon crystal is found at weak Landau level mixing and at an average density of $\bar{\rho}=9/16$, that is, near half-filling. We also find it at $\bar{\rho}=7/16$, but not at fillings farther away from $1/2$. At $\bar{\rho}=9/16$ at not-too-weak Landau level mixing, we find a continuous transition between the WM insulator and HC$_b$ on tuning the ratio of electronic interactions to kinetic energy. Assuming that the continuous transition survives quantum fluctuations, it will lie in the Wilson-Fisher Chern-Simons universality class~\cite{zhang1992chern}. This universality class may be obtained by tuning the following Lagrangian to criticality~\cite{quantized_footnote}, 
\begin{equation}
    \mathcal{L}= |D^{a'}_\mu \varphi|^2 - \mathcal{W}(|\varphi|^2 ) +\frac{\nu}{4\pi}(A-a')\textrm{d}(A-a'),
\end{equation}
where $D^a_\mu = \partial_\mu -\textrm{i}a_\mu$, $a\textrm{d}b=\epsilon^{\mu\nu\lambda}a^\mu \partial^\nu b^\lambda$ and $\mathcal{W}$ denotes the usual Wilson-Fisher potential. Our work thus also provides a candidate realization of this universality class.

Our analysis shows that both screening the Coulomb interaction and introducing a periodic potential favor the formation of Hall crystal phases. While we discussed results on the square lattice, Hall crystals on the triangular lattice can also be realized. Previous work obtained a supersolid with $\sqrt{3}\times \sqrt{3}$ charge order at half-filling on the triangular lattice for hard-core bosons in the extended Bose-Hubbard model with nearest neighbor interactions~\cite{melko2005supersolid, heidarian2005persistent, wessel2005supersolid} (the $\omega_c\rightarrow 0$ limit of Eq.~\ref{eq:ham_CS}). This crystal was obtained for arbitrarily large interactions due to geometric frustration in the classical problem. At finite $\omega_c$, we find that this $\sqrt{3}\times \sqrt{3}$ supersolid evolves into a Hall crystal. In fact, for this crystal, the three-body and current-density interactions are absent at mean-field level owing to the $C_3$ symmetry of the crystalline state. Since several moir\'e platforms realize a triangular lattice~\cite{wu2018hubbard, tang2020simulation, regan2020mott, ghiotto2021quantum} and have screened interactions, this suggests that they may be favorable platforms for the realization of Hall crystals, and their zero-field counterparts~\cite{dong2024anomalous, soejima2024anomalous, tan2024parent, sheng2024quantum, dong2024stability, desrochers2026energetics}.

\textit{Acknowledgments}\:---\:We thank J.~Dong, I.~Esterlis, P.~Kumar, L. Radzihovsky, A.~Reddy,  S.~Sondhi, Z.~D.~Shi and J.~Yu for helpful discussions. We especially thank S.~Kivelson for several useful comments. SB and SR are supported in part by the US Department of Energy, Office of Basic Energy Sciences, Division
of Materials Sciences and Engineering, under Contract
No.~DE-AC02-76SF00515. JMM is supported by a Leinweber Institute for Theoretical Physics fellowship.

\bibliography{draft.bib}

\newpage
\onecolumngrid

\newpage

\section*{Supplementary Material}

\tableofcontents

\vspace{10pt} 
\appendix

The supplementary material is organized as follows. In Sec.~\ref{derivation_bosonHam}, we derive the boson Hamiltonian in the continuum. In Sec.~\ref{mean-field_calc} and \ref{lattice sums}, we provide details of the mean-field analysis of the boson Hamiltonian using the Gutzwiller ansatz. In Sec.~\ref{sf_stiness}, we provide an expression for the superfluid stiffness using the Gutzwiller approximation. In Sec.~\ref{flux_attachment}, we explain how flux attachment is enforced in the anyon crystal. In Sec.~\ref{orbital_currents}, we derive the expression for the orbital currents evaluated at the center of each bond on the lattice. Finally, in Sec.~\ref{additional_numerical_results}, we provide additional numerical results at other lattice and magnetic fillings. 

\section{Derivation of the boson Hamiltonian}\label{derivation_bosonHam}

In this appendix, we derive the boson Hamiltonian by rewriting the statistical gauge field in terms of the boson degrees of freedom using the flux attachment constraint. We shall perform this calculation in the continuum and finally write a ultraviolet regularized lattice Hamiltonian based on the continuum Hamiltonian. 

The kinetic part of the Hamiltonian is,
\begin{equation}
    \mathcal{H}_{\textrm{kin}}(\bm{r})= \frac{1}{2m}\phi^\dagger (\bm{r})\left(-\im  \bm{\nabla}-\bm{a}'(\bm{r})\right)^2\phi(\bm{r}).
\end{equation}
Expanding the square, we obtain,
\begin{equation}
    \mathcal{H}_{\textrm{kin}}(\bm{r})=\frac{1}{2m}\phi^\dagger(\bm{r})(-\im \bm{\nabla})^2\phi(\bm{r})+\frac{1}{2m}\rho(\bm{r})a'^2(\bm{r})-\frac{1}{2m}\bm{j}^p(\bm{r})\cdot \bm{a}'(\bm{r}),
\end{equation}
where 
\begin{equation}
    \bm{j}^p=-\im \left( \phi^\dagger \bm{\nabla} \phi-\phi\bm{\nabla}\phi^\dagger\right),
\end{equation}
is the paramagnetic part of the boson current. By the paramagnetic current, we refer to part of the gauge-invariant current that does not depend on the gauge field $\bm{a}'$.  

We now write the second and third terms in momentum space, whose Hamiltonian density in momentum space we denote by $\mathcal{H}_1(\bm{k})$. This is given by,
\begin{equation}\label{H1_momentum_density}
    \mathcal{H}_1(\bm{k})=\frac{1}{2m}\int_q a'_i(\bm{k})a'^i(\bm{q})\rho(-\bm{k}-\bm{q}) -\frac{1}{2m}j_i^p(\bm{k})a'^i(-\bm{k}),
\end{equation}
where $\int_q=\int \textrm{d}^2 \bm{q}/{(2\pi)^2}$ and Einstein summation convention over the index $i=x,y$ is assumed. In the Coulomb gauge, we have the relationship, 
\begin{equation}\label{Coulomb_gauge_relationship}
    \bm{f}(\bm{k})=\im \frac{\hat{z}\times \bm{k}}{k}f_T(\bm{k}) \: \implies f_i(\bm{k})=\im \frac{\epsilon_{ij}k_j}{k}f_T(\bm{k}),
\end{equation}
where $\bm{f}$ is a vector valued function, $\epsilon_{ij}$ is the Levi-Civita symbol in two dimensions and the subscript $T$ denotes the transverse direction. Using Eq.~\ref{Coulomb_gauge_relationship}, we can rewrite Eq.~\ref{H1_momentum_density} as,
\begin{equation}    \mathcal{H}_1(\bm{k})=-\frac{1}{2m}\int_q \frac{k_iq_i}{kq}a_T(\bm{k}) a_T(\bm{q})\rho(-\bm{k}-\bm{q}) -\frac{1}{2m}j_T^p(\bm{k})a'_T(-\bm{k}).
\end{equation}
Now we use the flux attachment constraint,
\begin{equation}
    a_T(\bm{k})=\Phi\nu^{-1}\frac{\delta \rho(\bm{k})}{k}.
\end{equation}
Then, Eq.~\ref{H1_momentum_density} is,
\begin{equation}
  \mathcal{H}_1(\bm{k})=-\frac{4\pi^2}{2m} \nu^{-2}\int_q \frac{k_iq_i}{k^2q^2}\delta\rho(\bm{k})\delta\rho(\bm{q})\rho(-\bm{k}-\bm{q}) - \frac{1}{2m}\nu^{-1}j_T^p(\bm{k})\delta\rho(-\bm{k})\frac{2\pi}{k}.
\end{equation}
Note that,
\begin{equation}
    \rho(\bm{r})=\bar{\rho}-\delta\rho(\bm{r}) \implies \rho(\bm{k})=(2\pi)^2\bar{\rho}\delta^2(\bm{k})-\delta\rho(\bm{k}).
\end{equation}
Then the Hamiltonian is,
\begin{equation}
  \mathcal{H}_1(\bm{k})=\frac{\omega_c}{2\nu} \delta\rho(\bm{k})\delta\rho(-\bm{k})\frac{2\pi}{k^2}+\frac{4\pi^2}{2m} \nu^{-2}\int_q\frac{k_iq_i}{k^2q^2}\delta\rho(\bm{k})\delta\rho(\bm{q})\delta\rho(-\bm{k}-\bm{q}) - \frac{1}{2m}\nu^{-1}j_T^p(\bm{k})\delta\rho(-\bm{k})\frac{2\pi}{k}  .  
\end{equation}
It is instructive to write the Hamiltonian in real space. Using the inverse relationship of Eq.~\ref{Coulomb_gauge_relationship}, that is,
 $f_T(\bm{k})=-\im \epsilon_{ij}k_jf_i(\bm{k})/k=-\im (     \hat{z}\times \bm{k})\cdot \bm{f}(\bm{k})/k$, we can rewrite the last term of $\mathcal{H}_1(\bm{k})$ as,
 \begin{equation}
    \mathcal{H}_1(\bm{k})\supset -\frac{1}{2m} \nu^{-1}j_T^p(\bm{k})\delta\rho(-\bm{k})\frac{2\pi}{k} =     \im \frac{1}{2m}\nu^{-1}\frac{(\hat{z}\times \bm{k})\cdot \bm{j}^p(\bm{k})}{k}\delta\rho(-\bm{k})\frac{2\pi}{k} .
 \end{equation}
Fourier transforming back to real space, and writing the equivalent square lattice Hamiltonian, we obtain Eq.~\ref{eq:ham_CS} of the main text.

\section{Mean-field calculation}\label{mean-field_calc}

In this appendix, we describe the details of the mean-field calculation performed to obtaine the ground-states of the composite boson Hamiltonian using the Gutzwiller variational wavefunction. 

As stated in the main text, we consider a $4\times 4$ unit cell labeled as follows,
\begin{align}
    \begin{matrix}
        A &B& C& D\\
        E &F&G&H\\
        I&J&K&L\\
        M&N&O&P
    \end{matrix}
\end{align}
This unit cell supports charge configurations with period 4, and accommodates commensurate fillings so that $\bar{\rho}=p/16$, where $p\in \mathbb{Z}^+$. In the main text, we only present results for $\bar{\rho}=9/16$. 

First, one can compute the energy of the short-ranged part of the Hamiltonian $H_{\textrm{eff}}$. For the interactions, we obtain,
\begin{equation}
    E_{\textrm{short-range-int}}=\frac{U}{2}\sum_\alpha \rho_\alpha(\rho_\alpha-1)+V\sum_{\langle \alpha\beta\rangle }\rho_\alpha\rho_\beta
\end{equation}
From the kinetic energy we obtain,
\begin{equation}
    E_{\textrm{kin-energy}}=-2t\sum_{\langle \alpha\beta\rangle}\textrm{Re}[\langle \varphi_\alpha^*\rangle \langle \varphi_\beta\rangle ]
\end{equation}
where
\begin{equation}
    \langle \varphi_i\rangle = \sum_m c_{i,m-1}^*c_{im}\sqrt{m}
\end{equation}
and the density is,
\begin{equation}\label{eq_density}
    \rho_\alpha= \mel{\Psi}{\phi^\dagger_\alpha \phi _\alpha}{\Psi}=\sum_{m=0}^M m |c_{\alpha m}|^2. 
\end{equation}

In what follows, we resolve the long-ranged interactions onto the $4\times 4$ unit cell. 

\subsection{Two-body interaction}\label{two-body-int}
Any two-body interaction can be written as,   \begin{equation}\label{2b_unitcell}
E_{\textrm{cell}}^{2b}
=
\frac{1}{2}
\sum_{\alpha,\beta}
\delta\rho_\alpha\, \mathcal V_{\alpha\beta}\, \delta\rho_\beta.
\end{equation}
where $\mathcal{V}_{\alpha\beta}$ is given by,
\begin{equation}
    \mathcal{V}_{\alpha\beta}=\sum_{\bm{R}\in 4\mathbb{Z}^2}'V(|\bm{r}_{\alpha\beta}+\bm{R}|)
\end{equation}
where the prime over the sum indicates that we avoid the term corresponding to $\alpha=\beta$ and $\bm{R}=0$ in the sum. $\alpha,\beta=A,B,C,\hdots, P$. For our calculation, the two-body interactions in the logarithmic interaction, $V(\bm{r})=-(\omega_c/\nu)\ln(r)$. 

We first note that the pairwise interactions $\mathcal V_{\alpha\beta}$ can be classified into 
6 inequivalent classes. To see this, note that we can classify the values of $\bm{r}_\alpha-\bm{r}_\beta$ based on their distances. The six inequivalent classes are the following sets,
\begin{enumerate}
    \item $r_{\alpha\beta}=0$: $\bm{r}_{\alpha\beta}=\{(0,0)\}$.
    \item $r_{\alpha\beta}=1$: $\bm{r}_{\alpha\beta} =\{(1,0), (0,1), (3,0), (0,3)\}$.
    \item $r_{\alpha\beta}=\sqrt{2}$: $\bm{r}_{\alpha\beta} =\{(1,1), (3,3), (1,3), (3,1)\}$.
    \item $r_{\alpha\beta}=\sqrt{5}$: $\bm{r}_{\alpha\beta}=\{(1,2), (2,1), (2,3), (3,2)\}$. 
        \item $r_{\alpha\beta}=2$: $\bm{r}_{\alpha\beta}=\{(2,0), (0,2)\}$. 
    \item $r_{\alpha\beta}=2\sqrt{2}$: $\bm{r}_{\alpha\beta} =\{(2,2)\}$.  
\end{enumerate}

We therefore define the following sums,
\begin{align}
   \nonumber  &S_1= \mathcal{V}_{\alpha\beta}\big|_{r_{\alpha\beta}=0}\;\;\; ; \;\;\;S_2= \mathcal{V}_{\alpha\beta}\big|_{r_{\alpha\beta}=1}\;\;\; ; \;\;\;S_3= \mathcal{V}_{\alpha\beta}\big|_{r_{\alpha\beta}=\sqrt{2}}\\
    &S_4= \mathcal{V}_{\alpha\beta}\big|_{r_{\alpha\beta}=\sqrt{5}}\;\;\; ; \;\;\; S_5= \mathcal{V}_{\alpha\beta}\big|_{r_{\alpha\beta}=2}\;\;\; ;\;\;\; S_6= \mathcal{V}_{\alpha\beta}\big|_{r_{\alpha\beta}=2\sqrt{2}}
\end{align}

Using the fact that $\sum_\alpha \delta \rho_\alpha=0$, we can now write Eq.~\ref{2b_unitcell} as,
\begin{equation}\label{new_defn_unit_cell_2b}
    E_{\textrm{cell}}^{2b}=\frac{1}{2}\sum_{\alpha,\beta}\delta\rho_\alpha \Delta_{\alpha\beta}\delta\rho_\beta
\end{equation}
where $\Delta_{\alpha\beta}=S_i - S_1=\Delta_i$ if $(\alpha,\beta)$ belongs to the $i$-th class $(i=0,1,2,\hdots,6)$, and $\Delta_{\alpha\alpha}=\Delta_1=0$. Next, note that,
\begin{equation}
\Delta_{\alpha\beta}=\Delta_{\beta\alpha}
\end{equation}
due to the inversion symmetry of the potential. Also, note that the row and columns sums
\begin{equation}
    \sum_{\alpha}\Delta_{\alpha\beta}=\sum_\beta \Delta_{\alpha\beta}=\Lambda_0 = 4\Delta_2+4\Delta_3+4\Delta_4+2\Delta_5+\Delta_6
\end{equation}
are a constant independent of the row or column due to translational invariance of the potential. We can then write Eq.~\ref{new_defn_unit_cell_2b} as,
\begin{align}
    \nonumber E_{\textrm{cell}}^{2b}&=\frac{1}{2}\sum_{\alpha,\beta}[\bar{\rho}-\rho_\alpha]\Delta_{\alpha\beta}[\bar{\rho}-\rho_\beta]=\frac{1}{2}\left[16\Lambda_0\bar{\rho}^2-2\bar{\rho}\Lambda_0\sum_\beta\rho_\beta+\sum_{\alpha,\beta}\rho_\alpha\Delta_{\alpha\beta}\rho_\beta\right]\\
    &=\frac{1}{2}\sum_{\alpha,\beta}\rho_\alpha\Delta_{\alpha\beta}\rho_\beta-8\Lambda_0\bar{\rho}^2
\end{align}
This is a finite energy. The finite sums ($\Delta$s) are evaluated using the Ewald summation method in Sec.~\ref{lattice sums}.

\subsection{Three-body interaction}
The energy of the unit cell for the three-body interaction can be written as,
\begin{equation}\label{three-body}
    E_{\textrm{cell}}^{3b}=-\frac{\omega_c}{\nu}\frac{1}{2\pi\bar{\rho}}\sum_\alpha \delta\rho_\alpha \bm{\mathcal{A}}_\alpha^2
\end{equation}
where $\alpha=A,B,\hdots,P$. We therefore need to compute the density-dependent potential $\bm{\mathcal{A}}_\alpha$ for each $\alpha$ in the $4\times 4$ unit cell. We shall compute the potential defined by,
\begin{equation}
   \tilde{ \bm{\mathcal{A}}}_\alpha 
=
\sum_\beta \left[\delta\rho_\beta\sum_{\bm R\in 4\mathbb{Z}^2}'
\frac{\bm r_{\alpha\beta}+\bm R}{|\bm r_{\alpha\beta}+\bm R|^2}\right]=\sum_\beta \delta\rho_\beta \bm{K}(\bm{r}_{\alpha\beta})
\end{equation}
Note that $\bm{\mathcal{A}}_\alpha^2=\tilde{\bm{\mathcal{A}}}_\alpha^2$, and so it suffices to compute $\tilde{\bm{\mathcal{A}}}_\alpha$. 
\begin{theorem}The kernel $\bm{K}(\bm{r}_{\alpha\beta})=0$ whenever $\bm{r}_{\alpha\beta}\in \{(0,0), (2,0), (0,2), (2,2)\}$. 
\end{theorem}
\noindent \textit{Proof}. In component form, we can write,
\begin{equation}\label{explicit_eq}
    \bm{K}((m,n))=\sum_{p,q\in \mathbb{Z}}'\frac{(m+4p, n+4q)}{(m+4p)^2+(n+4q)^2}
\end{equation}
where the prime indicates that we skip the term in the sum where $(m,n)=(0,0)$ when $(p,q)=(0,0)$. Below, we will list transformations that the summand is odd under while the summation domain remains invariant. Then, upon such a transformation, all terms cancel pairwise and $\bm{K}((m,n))=0$.
\begin{enumerate}
\item For $(m,n)=(0,0)$, map $(p,q)\mapsto(-p,-q)$. 
\item For $(m,n)=(2,0)$, map $(p,q)\mapsto (-p-1, -q)$. 
\item For $(m,n)=(0,2)$, map $(p,q)\mapsto (-p, -q-1)$.
\item For $(m,n)=(2,2)$, map $(p,q)\mapsto (-p-1,-q-1)$. 
\end{enumerate}
\hfill $\square$

\begin{theorem}For the remaining $\bm{r}_{\mu\nu}$ in the set of  distance classes, the kernel depends on the following four lattice sums,
\begin{align}\label{three-body_sums} \alpha=\sum_{p,q\in\mathbb{Z}}\frac{1+4p}{(1+4p)^2+(4q)^2}\;\;\; ; \;\;\; \beta=\sum_{p,q\in\mathbb{Z}}\frac{1+4p}{(1+4p)^2+(1+4q)^2}\;\;\; ; \;\;\; \gamma=\sum_{p,q\in\mathbb{Z}}\frac{1+4p}{(1+4p)^2+(2+4q)^2}
\end{align}
Specifically, 
\begin{align}
   &\bm{K}((1,0))=(\alpha,0)\;\;\; ; \;\;\;\bm{K}((0,1))=(0,\alpha) \;\;\; ; \;\;\; \bm{K}((3,0))=(-\alpha,0)\;\;\; ; \;\;\;\bm{K}((0,3))=(0,-\alpha).\\
    &\bm{K}((1,1))=(\beta,\beta)\;\;\; ; \;\;\;\bm{K}((3,3))=(-\beta,-\beta) \;\;\; ; \;\;\; \bm{K}((1,3))=(\beta,-\beta)\;\;\; ; \;\;\;\bm{K}((3,1))=(-\beta,\beta).\\
    &\bm{K}((1,2))=(\gamma,0)\;\;\; ; \;\;\;\bm{K}((2,1))=(0,\gamma) \;\;\; ; \;\;\; \bm{K}((3,2))=(-\gamma,0)\;\;\; ; \;\;\;\bm{K}((2,3))=(0,-\gamma).
\end{align}
\end{theorem}
\noindent \textit{Proof}. Most of these follow from direct substitution in Eq.~\ref{explicit_eq}. Otherwise, one may perform a transformation such as $p\mapsto -p-1$ (or $q\mapsto -q-1$) and if needed interchange $p$ and $q$ indices to arrive at the rest. 
\hfill $\square$

The three sums $\alpha,\beta$ and $\gamma$ are computed in Sec.~\ref{lattice sums}. We can now explicitly write down $\tilde{\bm{\mathcal{A}}}_\alpha$ for each sublattice index. For instance,
\begin{align}
\nonumber \tilde{\bm{\mathcal{A}}}_A=
&
(\alpha,0)(\delta\rho_B-\delta\rho_D)
+(0,\alpha)(\delta\rho_E-\delta\rho_M)+
(\beta,\beta)(\delta\rho_F-\delta\rho_P)
+(\beta,-\beta)(\delta\rho_N
-\delta\rho_H)+\\
&
(\gamma,0)(\delta\rho_J-\delta\rho_L)
+(0,\gamma)(\delta\rho_G-\delta\rho_O)
\end{align}
This can be simplified as,
\begin{align}
    \tilde{\mathcal{A}}_A^x&= \alpha(\rho_D-\rho_B)+\beta(\rho_P-\rho_F+\rho_H-\rho_N)+\gamma(\rho_L-\rho_J)\\
    \tilde{\mathcal{A}}_A^y&=\alpha(\rho_M-\rho_E)+\beta(\rho_P-\rho_F-\rho_H+\rho_N)+\gamma(\rho_O-\rho_G)
\end{align}

The remaining $\tilde{\bm{\mathcal{A}}}_\alpha$s can be obtained by translation. Now,
\begin{equation}
        E_{\textrm{cell}}^{3b}=-g_2\sum_\alpha [\bar{\rho}-\rho_\alpha] \tilde{\bm{\mathcal{A}}}_\alpha^2=-g_2\bar{\rho}\sum_\alpha \tilde{\bm{\mathcal{A}}}_\alpha^2+g_2\sum_\alpha \rho_\alpha \tilde{\bm{\mathcal{A}}}_\alpha^2
\end{equation}
We compute these energies using Eq.~\ref{eq_density}.

\subsection{Current-density interaction}
The energy of the unit cell for the current-density interaction can be written as,
\begin{equation}\label{three-body}
    E_{\textrm{cell}}^{\textrm{cur-den}}=-\frac{\omega_c}{2\pi\bar{\rho}}\sum_\mu \bm{J}_\mu \cdot \bm{\mathcal{A}}_\mu
\end{equation}
where $\mu\in A,B,\hdots, P$. In the previous subsection, we computed $\tilde{\bm{\mathcal{A}}}_\mu$. Since $\bm{\mathcal{A}}_\mu=\hat{z}\times \tilde{\bm{\mathcal{A}}}_\mu$, we can write,
\begin{equation}
    \bm{\mathcal{A}}_\mu= \sum_\nu \delta\rho_\nu \bm{G}(\bm{r}_{\mu\nu})
\end{equation}
where $\bm{G}(\bm{r}_{\mu\nu})=\hat{z}\times \bm{K}(\bm{r}_{\mu\nu}).$ From the previous section we obtain,
\begin{align}
\nonumber \bm{\mathcal{A}}_A=
&
(0,\alpha)(\delta\rho_B-\delta\rho_D)
+(-\alpha,0)(\delta\rho_E-\delta\rho_M)+
(-\beta,\beta)(\delta\rho_F-\delta\rho_P)
+(\beta,\beta)(\delta\rho_N
-\delta\rho_H)+\\
&
(0,\gamma)(\delta\rho_J-\delta\rho_L)
+(-\gamma,0)(\delta\rho_G-\delta\rho_O)
\end{align}
which implies,
\begin{align}
    \mathcal{A}_A^x&= -\alpha(\rho_M-\rho_E)-\beta(\rho_P-\rho_F+\rho_N-\rho_H)-\gamma(\rho_O-\rho_G)\\
    \mathcal{A}_A^y&=\alpha(\rho_D-\rho_B)+\beta(\rho_P-\rho_F+\rho_H-\rho_N)+\gamma(\rho_L-\rho_J)
\end{align}
Again, the remaining $\bm{\mathcal{A}}_\alpha$s can be obtained by translation. 

Note that,
\begin{equation}
    J_\mu^\alpha =    \langle j_i^\alpha\rangle + \langle j_{i-e_\alpha}^\alpha\rangle 
\end{equation}
which using the Gutzwiller ansatz can be computed more explicitly. We find,
\begin{align}
\langle j_i^\alpha\rangle & = -\textrm{i}\left[\langle \varphi_i^\dagger\rangle \langle \varphi_{i+e_\alpha}\rangle -\langle \varphi^\dagger_{i+e_\alpha}\rangle \langle \varphi_i\rangle\right]= 2\textrm{Im}[\langle \varphi_i\rangle^*\langle \varphi_{i+e_\alpha}\rangle]
\end{align}
Therefore,
\begin{equation}
\langle J_i^\alpha\rangle = 2\textrm{Im}\left[\langle \varphi_i\rangle^*\langle \varphi_{i+e_\alpha}\rangle+\langle \varphi_{i-e_\alpha}\rangle^*\langle \varphi_i\rangle\right]
\end{equation}
This energy can also be minimized now using the sums in Eq.~\ref{three-body_sums}.

\section{Lattice sums}\label{lattice sums}

For the lattice sums involved in the previous section, we employ an Ewald summation to obtain finite answers. The idea of Ewald summation is to replace the total energy due to a long-range interaction on a lattice by the sum of a short-range part that can be summed efficiently in real space and a long-range part that can be summed efficiently in Fourier space. 

\subsection{Logarithmic interaction}\label{log_int_sums}
For the two-body logarithmic interaction, we wish to compute $\Delta_i$ $(i=1,2,\hdots, 6)$. These are defined as,
\begin{equation}
    \Delta_i  = S_i-S_1
\end{equation}
Let us denote $\Delta_i=\Delta(\bm{r})=S(\bm{r})-S_1$. Using the definitions in Sec.~\ref{two-body-int}, we find that, 
\begin{align}
    \Delta(\bm{r})&=-g_1\sum_{\bm{R}\in 4\mathbb{Z}^2}'\ln |\bm{r}+\bm{R}|+g_1\sum_{\bm{R}\in 4\mathbb{Z}^2}'\ln R\\
    =&-g_1\ln r - g_1\sum_{\bm{R}\in 4\mathbb{Z}^2\backslash \{0\}}\left[\ln |\bm{r}+\bm{R}|-\ln R\right]
\end{align}
with the understanding that $\Delta(0)=0$ and $g_1=\omega_c/\nu$. We set $g_1=1$ ahead, and will reintroduce it at the end. 

We now use the identity,
\begin{equation}
    \ln \frac{b}{a}= \frac{1}{2}\int_0^\infty \textrm{d} u\: \frac{e^{-ua^2}-e^{-ub^2}}{u}
\end{equation}
for $a,b>0$. Then, 
\begin{align}
    \Delta(\bm{r})&=\frac{1}{2}\int_0^\infty \textrm{d}u\:  \frac{e^{-ur^2}-e^{-u}}{u}- \frac{1}{2}\int_0^{\infty} \textrm{d}u\:\sum_{\bm{R}\in 4\mathbb{Z}^2\backslash \{0\}}\frac{e^{-u R^2}-e^{-u|\bm{r}+\bm{R}|^2}}{u}\\
    &=\frac{1}{2}\int_0^\infty \frac{\textrm{d} u}{u}\left[\sum_{\bm{R}\in 4\mathbb{Z}^2}(e^{-u|\bm{r}+\bm{R}|^2}-e^{-uR^2})+(1-e^{-u})\right]
\end{align}
We can now split the integral into a short-range and long-range part, denoted by $\Delta_{\textrm{sr}}(\bm{r})$ and $\Delta_{\textrm{lr}}(\bm{r})$. We introduce a splitting parameter $\Lambda$ and write $\Delta(\bm{r})$ as,
\begin{align}
    \Delta(\bm{r})&=\Delta_{\textrm{sr}}(\bm{r})+\Delta_{\textrm{lr}}(\bm{r})\\
    \Delta_{\textrm{lr}}(\bm{r})&=\frac{1}{2}\int_0^{\Lambda} \frac{\textrm{d} u}{u}\left[\sum_{\bm{R}\in 4\mathbb{Z}^2}(e^{-u|\bm{r}+\bm{R}|^2}-e^{-uR^2})+(1-e^{-u})\right]\\
    \Delta_{\textrm{sr}}(\bm{r})&=\frac{1}{2}\int_{\Lambda}^{\infty} \frac{\textrm{d} u}{u}\left[\sum_{\bm{R}\in 4\mathbb{Z}^2}(e^{-u|\bm{r}+\bm{R}|^2}-e^{-uR^2})+(1-e^{-u})\right]
\end{align}
The short ranged part is simplified in real space using the exponential integral defined as,
\begin{equation}
    E_1(z)=\int_z^\infty \textrm{d}t\: \frac{e^{-t}}{t}\;\;\; \textrm{for}\;\;\; |\textrm{arg}(z)|\leq \pi
\end{equation}
Then, we find that,
\begin{equation}\label{short_range}
    \Delta_{\textrm{sr}}(\bm{r})=\frac{1}{2}\left[E_1(\Lambda r^2)-E_1(\Lambda)+\sum_{\bm{R}\in 4\mathbb{Z}^2\backslash \{0\}}\left[E_1(\Lambda|\bm{r}+\bm{R}|^2)-E_1(\Lambda R^2)\right]\right]
\end{equation}
This can be computed efficiently numerically now. 

To compute the long-ranged part, we go to Fourier space using the Poisson resummation formula. The quantity to sum in momentum space is given by,
\begin{equation}\label{sigma_u_expression}
\Sigma_u(\bm{r})=\sum_{\bm{R}\in 4\mathbb{Z}^2}e^{-u|\bm{r}+\bm{R}|^2}
\end{equation}
Using the Poisson resummation formula, this can be written as,
\begin{equation}\label{poisson_sigma}
    \Sigma_u(\bm{r})=\frac{\pi}{16u}\sum_{\bm{K}}e^{-\frac{K^2}{4u}}e^{\textrm{i}\bm{K}\cdot \bm{r}}
\end{equation}
where $\bm{K}=\frac{\pi}{2}(m,n)$ and $m,n\in \mathbb{Z}$. Then,
\begin{align}
    \Delta_{\textrm{lr}}(\bm{r})&=\frac{1}{2}\int_{0}^{\Lambda} \frac{\textrm{d} u}{u}\left[\Sigma_u (\bm{r})-\Sigma_u (0)+(1-e^{-u})\right]\\
  &=\frac{1}{2}\int_0^{\Lambda}\textrm{d}u\: \frac{1-e^{-u}}{u}+\frac{\pi}{32}\sum_{\bm{K}\neq 0}(e^{-\textrm{i}\bm{K}\cdot \bm{r}}-1)\int_0^{\Lambda}\textrm{d}u \: u^{-2}e^{-\frac{K^2}{4u}}
\end{align}
Both of the integrals in $u$ can be done analytically now, and we obtain,
\begin{equation}\label{long_range}
        \Delta_{\textrm{lr}}(\bm{r})=\frac{1}{2}\ln \Lambda+\frac{\gamma}{2}+\frac{1}{2}E_1(\Lambda)+\frac{\pi}{8}\sum_{\bm{K}\neq 0} \frac{e^{-\frac{K^2}{4\Lambda}}}{K^2}(e^{\textrm{i}\bm{K}\cdot \bm{r}}-1)
\end{equation}
where $\gamma$ is the Euler-Mascheroni constant. 

We now compute $\Delta(\bm{r})$ for the displacements we are interested in, that is $\bm{r}=\{(0,0), (1,0),  (1,1), (1,2), (2,0), (2,2)\}$ using the sum of Eqs.~\ref{short_range} and ~\ref{long_range}, that is,
\begin{equation}\label{log_ewald}
    \Delta(\bm{r})=\frac{1}{2}\left[E_1(\Lambda r^2)+\sum_{\bm{R}\in 4\mathbb{Z}^2\backslash \{0\}}\left[E_1(\Lambda|\bm{r}+\bm{R}|^2)-E_1(\Lambda R^2)\right]\right]+\frac{1}{2}\ln \Lambda+\frac{\gamma}{2}+\frac{\pi}{8}\sum_{\bm{K}\neq 0} \frac{e^{-\frac{K^2}{4\Lambda}}}{K^2}(e^{\textrm{i}\bm{K}\cdot \bm{r}}-1)
\end{equation}
We find, inserting the factor of $g_1$,
\begin{align}
    \nonumber &\Delta_2/g_1 = 0.10126..\;\;\; ; \;\;\; \Delta_3/g_1= -0.162405..\;\;\; ; \;\;\;\Delta_4/g_1 = -0.339427..\\
    &\Delta_5/g_1= -0.249048..\;\;\;; \;\;\; \Delta_6/g_1= -0.422335.. 
\end{align}

\subsection{Three-body interaction}
In the three-body interaction, we are interested to compute,
\begin{equation}
\alpha=\sum_{p,q\in\mathbb{Z}}\frac{1+4p}{(1+4p)^2+(4q)^2}\;\;\; ; \;\;\; \beta=\sum_{p,q\in\mathbb{Z}}\frac{1+4p}{(1+4p)^2+(1+4q)^2}\;\;\; ; \;\;\; \gamma=\sum_{p,q\in\mathbb{Z}}\frac{1+4p}{(1+4p)^2+(2+4q)^2}
\end{equation}

To compute this, we will introduce the function,
\begin{equation}
    \bm{\Gamma}(\bm{r})=-\bm{\nabla}_r\Delta(\bm{r})=\sum_{\bm{R}\in 4\mathbb{Z}^2}'\frac{\bm{r}+\bm{R}}{|\bm{r}+\bm{R}|^2}
\end{equation}
Then,
\begin{equation}
\alpha = \Gamma_x((1,0))\;\;\; ; \;\;\; \beta = \Gamma_x((1,1))\;\;\; ; \;\;\; \gamma = \Gamma_x((1,2))
\end{equation}
We can now use Eq.~\ref{log_ewald} to compute $\bm{\Gamma}(\bm{r})$. Then,
\begin{align}
\bm{\Gamma}(\bm{r})&=\frac{1}{2}\left[-\bm{\nabla}_rE_1(\Lambda r^2)+\sum_{\bm{R}\in 4\mathbb{Z}^2\backslash \{0\}}\left[-\bm{\nabla}_rE_1(\Lambda|\bm{r}+\bm{R}|^2)\right]\right]+\frac{\pi}{8}\sum_{\bm{K}\neq 0} \frac{e^{-\frac{K^2}{4\Lambda}}}{K^2}\left(-\bm{\nabla}_re^{\textrm{i}\bm{K}\cdot \bm{r}}\right)\\
&=e^{-\Lambda r^2}\frac{\bm{r}}{r^2}+\sum_{\bm{R}\in 4\mathbb{Z}^2\backslash \{0\}}e^{-\Lambda|\bm{r}+\bm{R}|^2}\frac{\bm{r}+\bm{R}}{|\bm{r}+\bm{R}|^2}+\frac{\pi}{8}\sum_{\bm{K}\neq 0}\frac{e^{-\frac{K^2}{4\Lambda}}}{K^2}\bm{K}\sin(\bm{K}\cdot \bm{r})
\end{align}
Then, we compute these sums and find,
\begin{equation}
    \alpha =0.791276..\;\;\;  ;\;\;\; \beta = 0.327757..\;\;\; ; \;\;\; \gamma = 0.135761..
\end{equation}

\section{Superfluid stiffness}\label{sf_stiness}
In this appendix, we provide additional details of the computation of the superfluid stiffness using the Gutzwiller ansatz wavefunction. 

The superfluid stiffness is given using the formula,
\begin{equation}\label{phase_stiffness}
    \rho_{s,\alpha} = \frac{\partial^2\mathcal{E}(\theta_\alpha)}{\partial\theta_\alpha^2}\Bigg|_{\theta_\alpha=0},
\end{equation}
where $\mathcal{E}(\theta_\alpha)$ is the energy density given a uniform phase twist $\theta_\alpha$ along the $\alpha$-direction $(\alpha=x,y)$. As an illustration, we compute $\rho_{s,x}$. $\rho_{s,y}$ can be computed similarly. 

The phase-twisted Hamiltonian is,
\begin{equation}
    H=-\frac{t}{V}\sum_{\langle ij\rangle_x}[\varphi_i^\dagger \varphi_{j}e^{\im \theta}+\varphi_j^\dagger \varphi_i e^{-\im\theta}]
\end{equation} 
where $\langle ij\rangle_x$ denote the bonds along the $x$-direction and $\dots$ denote the terms in the Hamiltonian that do not depend on $\theta$. Naively, there is an additional contribution coming from the currents in the current-density interaction. However, this is an artifact of the lattice definition of the currents, and are hence ignored. 

Using the Gutzwiller wavefunction (with complex coefficients), we next compute the energy $E(\theta)$, given by,
\begin{centeq}
    E(\theta)=\mel{\Psi}{H}{\Psi} &=-\frac{t}{V}\sum_{\langle ij\rangle_x}\left[\langle \varphi_i^\dagger \rangle \langle \varphi_{j}\rangle  e^{\im\theta}+\langle \varphi_j^\dagger \rangle \langle \varphi_i\rangle e^{-\im\theta}\right] 
\end{centeq}
Now we use $\langle \varphi\rangle=|\varphi| e^{\im \vartheta}$, where both of these variables can be function of $\theta$. Then,
\begin{equation}
    E(\theta)=    -2\frac{t}{V}\sum_{\langle ij\rangle_x}|\varphi_i| |\varphi_j|\cos(\vartheta_j-\vartheta_i +\theta) 
\end{equation}
We shall now compute the bare stiffness, where we assume that $|\varphi|$ and $\vartheta$ do not depend on $\theta$. In that case, 
\begin{equation}
    \rho_{s,x}=\frac{1}{L^2}\left[2\frac{t}{V}\sum_{\langle ij\rangle_x}|\varphi_i| |\varphi_j|\cos(\vartheta_j-\vartheta_i)\right]
\end{equation}
where $L$ is the size of the system. Given the ground state values of the coefficients in the Gutzwiller ansatz, one can evaluate the stiffness using the above formula.

\section{Flux attachment in the anyon crystal}\label{flux_attachment}

In this appendix, we provide a partial explanation for the anyon crystal density and phase pattern obtained at weak LL mixing and small $t/V$ at $\bar{\rho}=9/16$. 

The AC is obtained at small $t/V$ and $\nu V/\omega_c$, and is thus a result of statistical interactions. It can thus be interpreted as the ground state of $H_{\textrm{st-int}}$, that is,
\begin{equation}
    H_{\textrm{st-int}}=-\frac{\omega_c}{\nu}\bigg[\sum_{i<j}\ln(r_{ij})\delta \rho_i\delta \rho_j+\frac{1}{2\pi\bar{\rho}}\sum_{i}\delta \rho_i (\bm{\mathcal{A}}_i)^2+ \frac{\nu}{2\pi\bar{\rho}}\sum_{i}\bm{J}_i \cdot \bm{\mathcal{A}}_i\bigg]
\end{equation}
Furthermore, within the mean-field approximation, the AC is obtained only for crystals that do not have inversion symmetry about every site in the unit cell. This can be a result of the three-body interaction, which favors anisotropic density patterns. The density of the anyon crystal can be parametrized by six parameters $A-F$, denoted in real space as,
\begin{equation}\label{density_ac}
   \langle \rho_i\rangle  =\begin{bmatrix}
        A & B & C & B\\
        B & D& E & D\\
        C & E& F& E\\
        B& D& E& D
    \end{bmatrix}_{mn\:\equiv\: \textrm{site}\: i}
\end{equation}
For small $t/V$, we find $A=C=1.0$ and $F=0.0$, and more generally $A\geq C> B> D> E> F$.

While the actual density pattern is the result of a complicated minimization problem, the fact that the density is parameterized by six parameters can be understood as follows. In the $4\times 4$ unit cell, upon choosing a site, say with density $A$, as the origin, the other sites can be grouped into classes as per their displacement from the origin, upto rotations and reflections of the square lattice. These form the same equivalence classes as those described in Sec.~\ref{two-body-int}. Assigning the same density to sites in the same class gives a six-parameter density pattern of the kind in Eq.~\ref{density_ac}.

The locations of the anyons are a result of flux attachment. Assuming that the densities follow the order $A\geq C> B> D> E> F$, and we find that the density is maximum and minimum at the density inversion centers $A$ and $F$ respectively. We will now show that due to flux attachment, the vortices are obtained at the sites of density extrema, which is suggested by the relation $\omega(\bm{r})\propto \delta \rho(\bm{r})$.  

This is most easily seen in the continuum. The kinetic energy of the Hamiltonian is,
\begin{equation}
    \mathcal{H}_{\textrm{kin}}(\bm{r})=\frac{1}{2m}\phi^\dagger(\bm{r})(-\im \bm{\nabla})^2\phi(\bm{r})+\frac{1}{2m}\rho(\bm{r})a'^2(\bm{r})-\frac{1}{2m}\bm{j}^p(\bm{r})\cdot \bm{a}'(\bm{r})
\end{equation}
Using $\phi = \sqrt{\rho}e^{\textrm{i}\theta}$, we can evaluate the phase-dependent part of the kinetic energy, which is given by,
\begin{equation}
      \mathcal{H}_{\textrm{kin}}(\bm{r})(\theta)=\frac{\rho}{2m}(\bm{\nabla}\theta)^2+\frac{\rho}{2m}\bm{a}'^2-\frac{\rho}{m}\bm{\nabla}\theta\cdot \bm{a}'=\frac{\rho}{2m}(\bm{\nabla}\theta-\bm{a}')^2
\end{equation}
Evaluating the saddle point, we obtain,
\begin{equation}
    \bm{\nabla}\theta=\bm{a}'
\end{equation}
Now, taking the curl on both sides, we obtain,
\begin{equation}
    \omega(\bm{r})=\bm{\nabla}\times \bm{\nabla}\theta(\bm{r}) =\bm{\nabla}\times \bm{a}'=\phi_0\nu^{-1}\delta\rho(\bm{r}) 
\end{equation}
where in the last equality we used the flux attachment constraint. This can be equivalently formulated on the lattice.

\section{Orbital currents}\label{orbital_currents}
In this appendix, we obtain an expression for the orbital currents in the HC and WC phases. The diamagnetic current is given by,
\begin{equation}
        \bm{j}^d(\bm{r})=-\bm{a}'(\bm{r})\rho(\bm{r})
\end{equation}
Rewriting the gauge potential using the density via the flux attachment constraint, we obtain, 
\begin{equation}
    (\bm{j}^d)_i(\bm{r})=-\frac{\phi_0}{2\pi \nu }\rho(\bm{r})\mathcal{\bm{A}}(\bm{r})
\end{equation}
where $\mathcal{\bm{A}}(\bm{r})= \int_r g(\bm{r}-\bm{r}')\delta\rho(\bm{r}')$ and $g(\bm{r})$ is the Biot-Savart kernel. On the lattice, this can be written as,
\begin{equation}
    (j_i^\alpha)^d=-\frac{\phi_0}{2\pi\nu}\rho_i \mathcal{A}^\alpha_i.
\end{equation}
Now, clearly for charge order that is inversion symmetric about every site, $\mathcal{\bm{A}}=0$ at the lattice sites. 

However, it not zero at the bonds of the square lattice. In particular, we shall evaluate the orbital current at the midpoint of the bonds. To do this, we first replace $\rho_i\mapsto \frac{1}{2}(\rho_i+\rho_{i+e_\alpha})$. Furthermore, we shall obtain an alternate expression for the gauge potential $\mathcal{A}(\bm{r})$, which may be interpreted as being evaluated at the center of the bond. 

Assume that the crystal has Fourier weight only in a finite number of modes. Then, we can write,
\begin{equation}
    \rho(\bm{r})=\bar{\rho}+\sum_\alpha \rho_{\bm{Q}_\alpha}\cos(\bm{Q}_\alpha\cdot \bm{r}+\vphi_\alpha)
\end{equation}
where $\bm{Q}_\alpha$ is the wavevector of the crystal. We can set $\vphi_\alpha=0$ if the crystal has an inversion center at some origin. This will be true for all the crystals we consider. Then,
\begin{equation}
    \mathcal{\bm{A}}(\bm{r})=-\sum_\alpha\rho_{\bm{Q}_\alpha}\int_{r'} g(\bm{r}-\bm{r}') \cos(\bm{Q}_\alpha\cdot \bm{r}')=\sum_\alpha \rho_{\bm{Q}_\alpha}\left[\frac{\hat{z}\times \bm{Q}_\alpha}{Q_\alpha^2}\sin(\bm{Q}_\alpha\cdot \bm{r})\right]
\end{equation}
This can be further written as,
\begin{equation}
    \mathcal{\bm{A}}(\bm{r})=\sum_\alpha \rho_{\bm{Q}_\alpha}\left[\frac{\hat{z}\times \bm{Q}_\alpha}{Q_\alpha^2}\sin(\bm{Q}_\alpha\cdot \bm{r})\right]=-\sum_\alpha \frac{\rho_{\bm{Q}_\alpha}}{Q_\alpha^2}\hat{z}\times \nabla \left[\cos(\bm{Q}_\alpha\cdot \bm{r})\right]
\end{equation}
Let us define the Fourier weight $\rho_\alpha = \rho_{\bm{Q}_\alpha}\cos(\bm{Q}_\alpha\cdot \bm{r})$ and, discretizing this equation for the lattice, we obtain the final current expression in the main text.

\section{Additional numerical results}\label{additional_numerical_results}

\begin{figure*}
    \centering
    \includegraphics[width=1\linewidth]{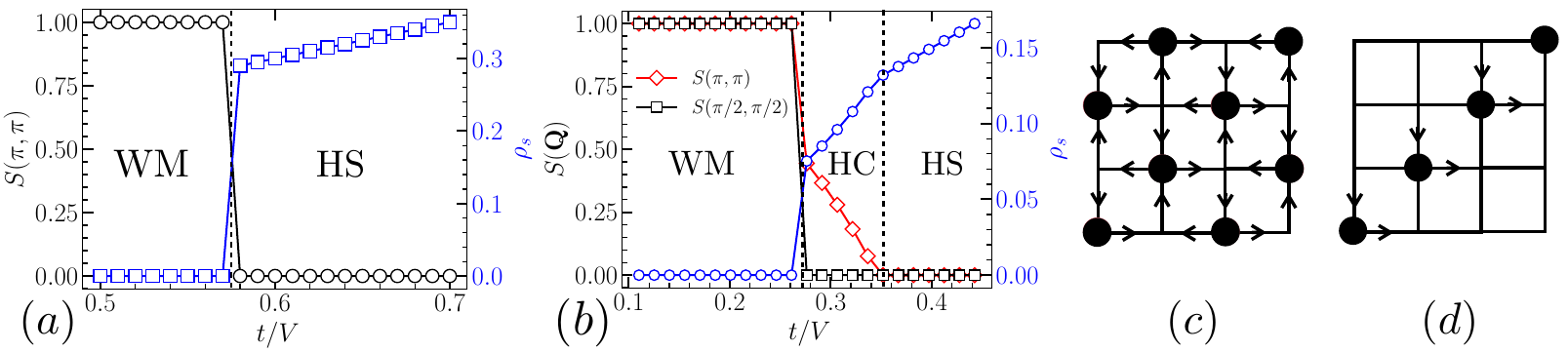}
    \caption{\textbf{(a)} Order parameters at $\bar{\rho}=1/2$. [$S(\pi,\pi)$ (black, circle), $\rho_s$ (blue, square)]. \textbf{(b)} Order parameters at $\bar{\rho}=1/4$. [$S(\pi,\pi)$ (red, diamond), $S(\pi/2,\pi/2)$ (black, square), $\rho_s$ (blue, circle)]. The structure factors are normalized with respect to their maximum value. \textbf{(c)} Checkerboard WM insulator at $\bar{\rho}=1/2$. \textbf{(d)} Diagonal striped WM insulator at $\bar{\rho}=1/4.$ Empty sites denote zero density, black circles denote unit density. The arrows on the links denote the direction of the orbital currents.}
\label{fig:crystal_orders_half_quarter}
\end{figure*}

In this appendix, we provide some additional numerical results. First, we show line cuts at fixed $\nu V/\omega_c$ for half and quarter lattice filling. These are obtained at $\nu V/\omega_c=1$ and $\nu=1/3$. At these fillings, at smaller $\nu V/\omega_c$, we do \textit{not} find an AC. 

At $\bar{\rho}=1/2$, we find a first order transition between a Hall state and a checkerboard WM insulator. At $\bar{\rho}=1/4$, we find a Hall state, a checkerboard HC and a diagonal striped WM insulator. The WM insulator and the HC are separated by a first order transition, while the HC and HS are separated by a continuous phase transition. 

\begin{figure}
    \centering
    \includegraphics[width=0.6\linewidth]{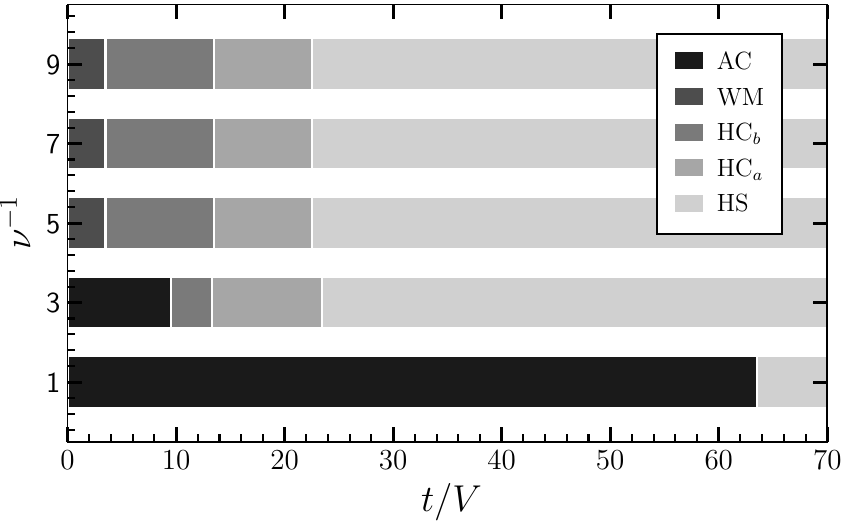}
    \caption{$\nu$ dependence of the $\bar{\rho}=9/16$ phases at fixed $\nu V/\omega_c=0.0033$.}
    \label{fig:nudependence}
\end{figure}

In Fig.~\ref{fig:nudependence}, we examine the dependence of the phase boundaries on $\nu$, at fixed $\nu V/\omega_c$. We choose $\bar{\rho}=9/16$ and $\nu V/\omega_c=0.0033$. At $\nu=1$, we only find two phases, a vortex lattice of composite bosons and the Hall state, separated by a first-order phase transition. This vortex lattice is an integer HC with two electrons and two holes. We have discussed the $\nu=1/3$ case before. For $\nu\leq 1/5$, the AC is replaced by the WC; furthermore, we find that the locations of the phase boundaries weakly move with $\nu$. Note that this does not necessarily imply that the AC is only obtained at $\nu=1/3$; for smaller values of $\nu V/\omega_c$, it may be possible to obtain the AC at smaller values of $\nu$ as well. However, this suggests that the AC in this parameter regime is most easily stabilized for the largest odd denominator LL filling fractions.

\end{document}